\documentclass{article}

\usepackage{amsthm, amsmath, amssymb, fullpage, authblk}
\usepackage[noend]{algorithm2e}
\usepackage{tikz, enumerate, comment}
\usetikzlibrary{shapes}
\usetikzlibrary{arrows}
\usetikzlibrary{decorations.pathreplacing, decorations.pathmorphing}

\title{Constant Factor Approximation\\ for Capacitated $k$-Center with
Outliers\footnote{This work is partially supported by Foundation for Polish Science grant HOMING PLUS/2012-6/2.} }

\author{Marek Cygan}
\author{Tomasz Kociumaka}
\affil{Institute of Informatics, University of Warsaw, Poland\\
  \texttt{[cygan, kociumaka]@mimuw.edu.pl}}

\date{}

\newcommand{\F}{\mathcal{F}}
\newcommand{\C}{\mathcal{C}}
\newcommand{\pint}{\mathbb{Z}_{\ge 0}}
\newcommand{\preal}{\mathbb{R}_{\ge 0}}
\newcommand{\sub}{\subseteq}
\newcommand{\floor}[1]{\left\lfloor #1 \right\rfloor}

\newcommand{\sm}{\setminus}
\newcommand{\fullsup}{\textsc{Capacitated} $k$-\textsc{supplier with Outliers}}
\DeclareMathOperator{\argmax}{argmax}
\newcommand{\fullcen}{\textsc{Capacitated} $k$-\textsc{center with Outliers}}

\theoremstyle{plain}
\newtheorem{theorem}{Theorem}
\newtheorem{lemma}[theorem]{Lemma}
\newtheorem{corollary}[theorem]{Corollary}
\newtheorem{fact}[theorem]{Fact}

\theoremstyle{definition}
\newtheorem{definition}[theorem]{Definition}

\newcommand{\defproblem}[4]{
  \vspace{1mm}
\noindent\fbox{
  \begin{minipage}{0.96\textwidth}
  #1 \\
  {\bf{Input:}} #2  \\
  {\bf{Find:}} #3 
  {\bf{Minimize:}} #4
  \end{minipage}
  }
  \vspace{1mm}
}

\begin{document}

\maketitle

\begin{abstract}
The $k$-center problem is a classic facility location problem,
where given an edge-weighted graph $G = (V,E)$ one is to find a subset 
of $k$ vertices $S$, such that each vertex in $V$ is ``close''
to some vertex in $S$. 
The approximation status of this basic problem is well understood,
as a simple $2$-approximation algorithm is 
known to be tight. 
Consequently different extensions were studied.

In the capacitated version of the problem each vertex is assigned
a capacity, which is a strict upper bound on the number of clients
a facility can serve, when located at this vertex.
A constant factor approximation for the capacitated $k$-center was obtained
last year by Cygan, Hajiaghayi and Khuller~[FOCS'12], which
was recently improved to a $9$-approximation
by An, Bhaskara and Svensson~[arXiv'13].

In a different generalization of the problem some
clients (denoted as outliers) may be disregarded.
Here we are additionally given an integer $p$ and the 
goal is to serve exactly $p$ clients, which the algorithm
is free to choose.
In 2001 Charikar et al. [SODA'01] presented a $3$-approximation
for the $k$-center problem with outliers.

In this paper we consider a common generalization of the two
extensions previously studied separately, i.e.\ we work
with the capacitated $k$-center with outliers. We present the first constant factor approximation algorithm
with approximation ratio of $25$ even for the case of non-uniform hard capacities.
\end{abstract}

\section{Introduction}

The $k$-center problem is a classic facility location problem and
is defined as follows: given a finite set $V$ and a symmetric distance (cost)
function $d: V\times V \to \preal$ satisfying the triangle inequality,
find a subset $S \subseteq V$ of size $k$ such that each vertex
in $V$ is ``close'' to some vertex in $S$. More formally, once we choose $S$ 
the objective function to be minimized is $ \max_{v \in V} \min_{u \in S}
d(v,u)$.
The vertices of $S$ are called \emph{centers} or \emph{facilities}.
The problem is known to be NP-hard \cite{GJ}.  Approximation algorithms for the 
$k$-center problem have been well studied and are known to be optimal \cite{Gonzalez,HS1,HS2,HN}. 

In the capacitated setting, studied for twenty years already,
we are additionally given a capacity function
$L:V \to \pint$ and no more than $L(u)$ vertices (called \emph{clients}) may be assigned
to a chosen center at $u \in V$.
For the special case when all the capacities are {\em identical} 
(denoted as the \emph{uniform} case), a $6$-approximation was developed by
Khuller and Sussmann~\cite{KS} improving the previous bound of $10$ by Bar-Ilan, Kortsarz and Peleg \cite{BKP}.
In the \emph{soft} capacities version, in contrast to the standard (\emph{hard} capacities), we are
allowed to open several facilities in a single location, i.e.\ the facilities may
form a multiset.
For the uniform soft capacities version the best known approximation ratio equals $5$~\cite{KS}.
For general hard capacities a constant factor approximation
has been obtained only recently~\cite{chk-focs12}, somewhat surprisingly
by using LP rounding.
It was followed by a cleaner and simpler approach of An, Bhaskara and Svensson~\cite{svensson} 
who gave a $9$-approximation algorithm.
From the hardness perspective a $(3-\varepsilon)$ lower bound
on the approximation ratio is known~\cite{JACm,chk-focs12}.

Another natural direction in generalizing the problem
is an assumption that instead of serving all the clients
we are given an integer $p$ and we are to select
exactly $p$ clients to serve. The disregarded clients are
in the literature called \emph{outliers}.
The $k$-center problem with outliers admits a $3$-approximation algorithm,
which was obtained by Charikar et al.~\cite{CKMN}.

In this article we study a common generalization of the two mentioned variants 
of the $k$-center problem, i.e.\ involving both capacities and outliers.
In order to simplify our algorithms we work with a slight generalization,
the \textsc{Capacitated} $k$-\textsc{supplier with Outliers} problem,
where vertices are either clients or potential facility
locations. These vertices may coincide, so that one may have both a client and a
potential facility location at the same point, as in $k$-\textsc{center}. 
Below we give the formal problem definition. 

\defproblem{\textsc{Capacitated} $k$-\textsc{supplier with Outliers}}
{Integers $k,p\in \pint$, finite sets $\C$ and $\F$, a symmetric distance
(cost) function $d : (\C \cup \F) \times (\C \cup \F) \to \preal$ satisfying 
  the triangle inequality, and a capacity function $L: \F \to \pint$} {
Sets $C\sub \C$, $F\sub \F$, and a function $\phi: C \to F$ satisfying
\begin{itemize}
  \item $|C|=p$,
  \item $|F|=k$,
  \item $|\phi^{-1}(u)|\le L(u)$ for each $u\in F$.
\end{itemize}} 
{$\max_{v\in C} d(v, \phi(v))$.}

Again, in the \emph{soft} capacities version, $F$ is allowed to be a multiset,
and in the \emph{uniform} capacities version, the capacity function $L$ is
constant.

Existence of an $r$-approximation algorithm for \textsc{Capacitated} $k$-\textsc{center with Outliers}
can be shown to be equivalent to existence of an $r$-approximation algorithm for \textsc{Capacitated} $k$-\textsc{supplier with Outliers}
(see Appendix~\ref{app:equivalence}).
Interestingly, such an equivalence is not known to hold if we do not allow outliers:
the best known approximation factor for the \textsc{Capacitated} $k$-\textsc{supplier} is 11
while for the \textsc{Capacitated} $k$-\textsc{center} it is 9, see \cite{svensson}.

\subsection{Our results and organization of the paper}

The following is the main result of this paper.

\begin{theorem}
The \fullsup\ 
problem, both in hard and soft
capacities version, admits a 25-approximation algorithm.
The hard uniform capacities version admits a 23-approximation, and soft uniform
capacities -- a 13-approximation.
\end{theorem}
Note that taking $\C=\F=V$ shows that 
the $k$-supplier problem generalizes the $k$-center problem, 
and consequently gives the same approximation bounds for the latter.
\begin{corollary}
The \fullcen\ 
problem, both in hard and soft
capacities version, admits a 25-approximation algorithm. 
The hard uniform capacities version admits a 23-approximation, and soft uniform
capacities -- a 13-approximation.
\end{corollary}

It is worth noting, that the already known approximation algorithm
for the $k$-center problem with outliers relies on the fact 
that a single vertex can serve all the clients that are its neighbors,
i.e.\ there are no capacity constraints.
At the same time the previous approximation algorithms
for the capacitated $k$-center problem (both in the uniform
and non-uniform case) heavily used the fact that 
each vertex of the graph is close to some center in any solution.
For this reason it was possible to create a path-like~\cite{chk-focs12}
or tree-like~\cite{svensson} structure
with integrally opened non-leaf vertices, that was the crux in
the rounding process. 
Consequently none of the algorithms for the two previously independently
studied extensions of the basic problem, i.e.\ capacities and outliers,
works for the problem we are interested in.

The first step of our algorithm (Section~\ref{sec:thresholding}) is the standard
thresholding technique, where we reduce a general metric to a distance metric of an unweighted graph.
In Section~\ref{sec:skeleton} we introduce our main conceptual
contribution, i.e.\ the notion of a {\em skeleton}.
A skeleton is a set $S$ of vertices, for which
there exists an optimum solution $F \subseteq \F$, 
such that each vertex of $S$ can be injectively mapped to a
nearby vertex of $F$ and moreover each vertex of $F$
is close to some vertex of $S$.
Intuitively a skeleton is not yet a solution, but it looks
similar to at least one optimum solution.
If no outliers are allowed, any inclusion-wise maximal subset of $\F$ with vertices
far enough from each other, is a skeleton.  In~\cite{chk-focs12} and \cite{svensson},
such a set is then mapped to non-leaf vertices of the structure steering the rounding process.
We use a skeleton in a similar way, but before we are able to do that,
we need to bound the integrality gap. Without outliers, it was sufficient to take the
 standard LP relaxation and decompose the graph into connected components. 
Although with outliers this is no longer the case,
as shown in Section~\ref{sec:clustering}, a skeleton lets us both strengthen the LP relaxation, adding an appropriate constraint,
and obtain a more granular decomposition of the initial instance into several
subinstances, for which the strengthened LP relaxation is feasible and has bounded integrality gap.
Further in Section~\ref{sec:rounding} we show how each of these
smaller instances can be independently rounded using tools previously applied for the capacitated setting~\cite{svensson}.\footnote{The final
rounding step can be also done using the path-like structures notion of~\cite{chk-focs12},
however we use the ideas of~\cite{svensson} as it allows cleaner presentation.}
Section~\ref{sec:wrap-up} contains a wrap-up of the whole algorithm.
The improvements in the approximation ratio when soft or uniform capacities are considered, are presented in Appendix~\ref{app:improvements}. 

\subsection{Related facility location work}

The {\em facility location} problem is a central problem in operations
research and computer science and has been a testbed for many new
algorithmic ideas resulting a number of different
approximation algorithms. 
In this problem,  given a metric (via a weighted graph $G$), a set of nodes called {\em clients},  and opening costs on some nodes called {\em facilities}, the goal is to open a subset of facilities such that the sum of their opening costs and connection costs of clients to their nearest open facilities is minimized. 
Up to now, the best known approximation ratio is 1.488, due to Li~\cite{Li11}
who used a randomized selection in Byrka's algorithm \cite{Byr07}.  Guha and Khuller~\cite{GK} showed that this problem is hard to approximate within a factor better than 1.463, assuming $NP
\not\subseteq DTIME\big[n^{O(\log \log n)}\big]$.

When the facilities have capacities, the problem is called the {\em capacitated
facility location} problem. It has also received a great deal of
attention in recent years. Two main variants of the problem are
soft-capacitated facility location and hard-capacitated facility
location: in the latter problem, each facility is either opened at
some location or not, whereas in the former, one may specify any
integer number of facilities to be opened at that location. Soft
capacities make the problem easier and by modifying approximation
algorithms for the uncapacitated problems, we can also handle this
case~\cite{STA,JV}.
To the best of our knowledge all the existing constant-factor approximation
algorithms for the general case of hard capacitated facility location
are local search based, and the most recent of them is the $5$-approximation
algorithm of Bansal, Garg and Gupta~\cite{bgg12}.
The only LP-relaxation based approach for this problem is due to Levi, Shmoys
and Swamy~\cite{LSS04} who gave a 5-approximation algorithm for the
special case in which all facility opening costs are equal
(otherwise the LP does not have a constant integrality gap).
Obtaining an LP based constant factor approximation algorithm
for capacitated facility location is considered a major problem in
approximation algorithms~\cite{sw-book}.

A problem very close to both facility location and $k$-center is the {\em $k$-median} problem in which we want to open at most $k$ facilities 
and the goal is to minimize the sum of connection costs of clients to their nearest open facilities.
Very recently Li and Svensson~\cite{li-svensson} obtained an LP rounding $(1+\sqrt{3})$-approximation algorithm,
improving upon the previously best $(3+\varepsilon)$-approximation local search
algorithm of Arya et al.~\cite{kmed3}.
Unfortunately obtaining a constant factor approximation algorithm for capacitated $k$-median still remains open despite consistent effort. 
The only previous attempts with  constant approximation factors for this problem violate the capacities within a constant 
factor for the uniform capacity case~\cite{CGTS} and the non-uniform capacity case~\cite{CR05} or exceed the number $k$ of facilities by a constant factor~\cite{BCR01}. 

\section{Preliminaries}

For a fixed instance of the \fullsup, we call $(C,F,\phi)$
\emph{a solution} if it satisfies the required conditions. We often identify the
solution by $\phi$ only (considering it as a partial function from $\C$ to $\F$),
using $C_\phi$ and $F_\phi$ to refer to the other elements of the triple.
If $\phi$ satisfies $\max_{v\in C} d(v, \phi(v))\le \tau$, we say that $\phi$ is
a distance-$\tau$ solution.

Let $G=(V,E)$ be an undirected graph. By $d_G$ we denote the metric defined by
$G$.
For sets $A,B\sub V$ we define $d_G(A,B) = \min_{a\in A, b\in
B} d_G(a,b)$. If $B=\{b\}$ we write $d_G(A,b)$ instead of $d_G(A,B)$.

For a vertex $v\in V$ and an integer $k\in \pint$ we denote
$N^k_G(v)=\{u \in V : d_G(u,v)=k\}$ and $N^k_G[v]=\{u \in V : d_G(u,v)\le k\}$.
We omit the superscript for $k=1$ and the subscript if there is no confusion
which graph we refer to.

For a set $S$ and an element $s$ by $S+s$ we denote $S\cup \{s\}$.

\section{Reduction to graphic instances}
\label{sec:thresholding}

As usual when working with a min max problem we start with the standard
thresholding argument, i.e.\ reduce a general metric function to a metric defined by an unweighted graph. 

We say that an instance of the $k$-supplier problem is \emph{graphic},
if $d$ is defined as the distance function of an unweighted bipartite graph
$G=(\C,\F,E)$, and the goal is to find a distance-1 solution. 
An $r$-approximation algorithm is then allowed to either give a distance-$r$
solution, or, only if it finds out that no distance-1 solution
exists, a NO answer.

Below we show how to build an $r$-approximation algorithm for \fullsup\ given an
$r$-approximation (in the aforementioned sense) for the graphic instances.
Correctness of the reduction is standard. If an optimal solution exists, then
its value $OPT$ belongs to $T$. In particular, in the phase corresponding to $OPT$,
 there is a distance-1 solution in $G_{\le OPT}$. Thus the algorithm for
graphic instances is required to find a solution. Therefore returns a
solution $\phi$ for the first time at phase corresponding to $\tau^*\le OPT$.
Since $d(v,u)\le\tau^* d_{G_{\le \tau^*}}(v,u)$, $\phi$ is a distance-$r\cdot\tau^*$
solution, hence also distance-$r\cdot OPT$ solution.

\begin{algorithm}
\caption{Reduction to graphic instances}\label{alg:red}
$T:=\{d(v,u) : v\in \C, u\in \F\}$\;
\ForEach{$\tau\in T$ in ascending order}{
	$G_{\le \tau} := (\C,\F,\{(v,u): d(v,u)\le \tau\})$\;
	solve the graphic instance for $G_{\le \tau}$\;
	\lIf{a solution $\phi$ found}{\Return $\phi$\;}
}
\Return {NO}\;
\vspace{.2cm}
\end{algorithm} 
\section{Finding a skeleton}
\label{sec:skeleton}
From now on we work with graphic instances only.
Without loss of generality we may assume that $L(u)\le deg(u)$ for each $u\in
\F$.  Indeed, setting
 $L(u) := \min(L(u),deg(u))$ has no influence on distance-1 solutions, while no
 additional distance-$r$ solutions are created.

The first phase of the algorithm outputs several subsets of $\F$.
If a distance-1 solution exists, at least one of them resembles (in a certain
sense, to be defined later) a distance-1 solution and can be successfully used
by the subsequent phases as a hint for constructing a distance-$r$ solution. 
We formalize the features of a good hint in the following definition.
\begin{definition}
A set $S \sub \F$ is called a \emph{skeleton} if
\begin{itemize}
  \item ({\bf separation property}) $d(u,u')\ge 6$ for any $u,u'\in S$, $u \neq u'$,
  \item there exists a distance-1 solution $(C_\phi, F_\phi, \phi)$ such that:
  \begin{itemize}
    \item ({\bf covering property}) $d(u, S)\le 4$ for each $u\in F_\phi$,
    \item ({\bf injection property}) there exists an injection $f: S \hookrightarrow F_\phi$ 
    satisfying $d(u, f(u))\le2$ for each $u \in S$.
  \end{itemize}
\end{itemize}
If just separation and injection properties are satisfied, we call $S$ a \emph{preskeleton}.
\end{definition}

In other words a skeleton is a set $S$, each vertex of which
can be injectively mapped to a vertex of a distance-1 solution $F_\phi$,
and at the same time no two vertices of $S$ are close
and $N^4[S]$ contains the whole set $F_\phi$.

Note that the separation property implies that sets $N^2[u]$ are pairwise disjoint for
$u\in S$, hence any function $f:S \to F_\phi$ satisfying $d(u, f(u))\le2$ 
is in fact an injection, however we make it explicit for the sake of presentation.

\begin{lemma}\label{lem:ind}
Let $S$ be a preskeleton and let $U=\{u\in \F : d(u,S)\ge 6\}$.
Then $S$ is a skeleton, or $U\ne \emptyset$ and $S+s$ is a preskeleton,
where $s$ is a highest-capacity vertex of $U$.
\end{lemma}

\begin{proof}
Let $\phi$ be a distance-1 solution, which witnesses $S$ being a preskeleton,
where $f : S \hookrightarrow F_\phi$ satisfies the injection property.
If $\phi$ witnesses $S$ being a skeleton, we are done.
Otherwise the covering property is not satisfied,
hence there exists $u\in F_\phi$ such that $d(u,S)>4$.
Since $d$ is a distance function of a bipartite graph, 
this implies $d(u,S)\ge 6$, so $u\in U \ne \emptyset$.
If $|F_\phi \cap N^2[s]|\ge 1$, then $\phi$ already witnesses $S+s$ being a
preskeleton, as one can extend the injection $f$ by mapping 
a vertex of $F_\phi \cap N^2[s]$ to $s$.
Therefore, we may assume that $N^2[s]\cap F_\phi = \emptyset$.
In particular, this means that the clients in $N(s)$ are not served by any
facility of $F_\phi$.

Let us modify $\phi$ to obtain $\psi$ as follows:
close the facility in $u$, opening one in $s$ instead.
Let $c$ be the number of clients assigned to $u$ in $\phi$.
No longer serve these, instead serve any $c$ neighbors of $s$ in
$\psi$ (as we have observed before, they are not served in $\phi$).
Note that $c\le L(u)\le L(s)\le \deg(s)$ by the choice of $u$ maximizing the
capacity and by the assumption of $L$ being bounded by $\deg$.
Consequently, there are enough neighbors of $s$ to serve, and the capacity
constraint for $s$ is satisfied.
Moreover, the number of open facilities and the number of served clients are
preserved.
Other open facilities remain unchanged, so $\psi$ satisfies
the capacity and distance constraints for them, and therefore is a distance-1 solution.
Finally, consider a function $f'=f + (s, s)$.
As $s$ is at distance at least $6$ from $S$, by the 
injection property for $S$ we know that $s$ does not belong to the image of $f$,
hence $f'$ is an injection.
Consequently $\psi$ and $f'$ ensure $S+s$ satisfies the injection property.
Moreover $s$ is far from $S$, hence $S+s$ is a preskeleton.
\end{proof}

With $\emptyset$ being trivially a preskeleton provided that any distance-1
solution exists, Lemma~\ref{lem:ind} lets us generate a sequence of
sets, which contains a skeleton (see Algorithm~\ref{alg:pre}).
Note that any skeleton, by the injection property, is of size at most $k$.

\begin{lemma}
If there exists a distance-1 solution, there is at least one skeleton among sets
output by Algorithm~\ref{alg:pre}.
\end{lemma}
\begin{algorithm}
\caption{Construction of a family of sets containing at least one skeleton.}\label{alg:pre}
$S := \emptyset$\;
\While{$|S|\le k-1$}{
	$U := \{u\in \F \; : \; d(u,S) \ge 6\}$\;
	\lIf {$U = \emptyset$}{\KwSty{break}}\;
	$s := \textrm{argmax}\{ L(u) \; : \; u \in U\}$\;
	$S := S + s$\;
	\KwSty{output} $S$\;
}
\vspace{.2cm}
\end{algorithm} 

\section{Clustering}
\label{sec:clustering}

For a set $S \subseteq \F$ define the following linear program
$LP_{k,p}(G,L,S)$, where a variable $y_u$ for $u \in \F$ denotes whether
we open a facility in $u$ or not, while
a variable $x_{uv}$ for $u \in \F$, $v \in \C$
corresponds to whether $u$ serves $v$ or not.

\begin{alignat}{9}
\label{lp:1} \sum_{u\in \F} y_u &= k & &\\
\label{lp:2}\sum_{u\in \F, v\in \C} x_{uv} &= p & &\\
\label{lp:3}x_{uv} &\le y_u & & \text{for each $u\in \F, v\in \C$}\\
\label{lp:4}\sum_{v} x_{uv} & \le L(u)\cdot y_u \quad & & \text{for each $u\in \F$}\\
\label{lp:5}\sum_{u} x_{uv} & \le 1 & & \text{for each $v\in \C$} \\
\label{eq:near}\sum_{u\in \F \cap N^2[s]} y_{u} &\ge 1 & & \text{for each $s\in S$}\\
\label{eq:forb}x_{uv} &= 0 & &\text{for each $u\in \F, v\in \C$ such that $(v,u)\notin E$}\\
\mathbf{0} \le x,y &\le \mathbf{1} & &
\end{alignat}

Constraints $(\ref{lp:1})-(\ref{lp:5}),(\ref{eq:forb})$
are the standard constraints for \fullsup,
ensuring that we open exactly $k$ facilities (\ref{lp:1}), 
serve exactly $p$ clients (\ref{lp:2}),
obey capacity constraints (\ref{lp:3})-(\ref{lp:5}),
and serve clients which are close to facilities (\ref{eq:forb}).

Observe that if $S$ is a skeleton and a distance-1 solution $\phi$ witnesses
that fact, we get a feasible solution of $LP_{k,p}(G,L,S)$ setting $y_u = 1$
iff $u\in F_\phi$ and $x_{uv}=1$ iff $v\in C_\phi$ and $v=\phi(u)$.
Indeed the injection property ensures that constraint (\ref{eq:near}) is satisfied.
However, as usual in a capacitated problem with hard constraints, 
the integrality gap of this LP is unbounded.
Similarly to the standard \textsc{capacitated} $k$-\textsc{center}~\cite{chk-focs12},
this issue is addressed by considering the connected components of $G$
separately. 
When all the clients need to be served having a connected graph
with a feasible solution of the standard LP is enough to round it~\cite{svensson,chk-focs12}.
However, if we allow outliers, there are sill connected instances with arbitrarily large integrality gap (a simple construction
is presented in Appendix~\ref{app:example}).
For this reason we use the additional constraint~(\ref{eq:near})
together with the assumption that all the vertices are close to $S$.
This way we crucially exploit the 
covering, injection and separation properties of a skeleton.

In the following we shall prove that any instance 
with a skeleton can be decomposed into several smaller instances
with additional properties. In the next section we will show how
to round the obtained smaller instances.

\begin{lemma}\label{lem:dp}
Let $S\sub \F$, let $G_1,\ldots,G_\ell$ be components
of $G$ after all vertices $v$ with $d(v,S)>5$ are removed
and let $S_i=S\cap V(G_i)$ for $1 \le i \le \ell$.

If $S$ is a skeleton, then in polynomial time one can find partitions $k=\sum_{i=1}^\ell k_i$
and $p=\sum_{i=1}^\ell p_i$ such that $LP_{k_i,p_i}(G_i,L,S_i)$ are all feasible.
\end{lemma}

\begin{proof}
Observe that if $S$ is a skeleton, then a witness solution $\phi$
opens facilities at distance at most 4 from $S$, and thus serves clients
with distance at most 5 from $S$. Consequently all vertices further from $S$ can
be safely removed and $S$ remains a skeleton.
Then $G$ might contain several connected components $G_1,\ldots,G_\ell$
with $G_i=(\C_i,\F_i,E_i)$.  The witness solution
$\phi$ can be partitioned among these components so that we get assignments
$\phi_i$ which in total open $k$ facilities to serve $p$ clients. In particular,
this means that for some partitions $k = \sum_{i} k_i$ and $p = \sum_{i} p_i$
sets $S_i=S\cap \F_i$ are skeletons, and consequently
$LP_{k_i,p_i}(G_i,L,S_i)$ are feasible. The latter condition can be
tested efficiently for any values $k_i$ and $p_i$.
While we cannot exhaustively test all partitions of $k$ and $p$, dynamic
programming lets us find partitions such that these linear programs are
feasible for each $i$.

For $i\in \{0,\ldots,\ell\}$, $k'\in\{0,\ldots,k\}$ and  $p'\in\{0,\ldots,p\}$
define a boolean value $F[i][k'][p']$, which equals {\bf true} iff there exist partitions
$k'=\sum_{j=1}^i k_j$ and $p' = \sum_{j=1}^i p_j$ such that
$LP_{k_j,p_j}(G_j,L,S_j)$ are all feasible for $j\le i$.

Clearly $F[0][0][0]$ is true, while
$F[0][k'][p']$ is false for any other pair $(k',p')$.
For $i > 1$ the value $F[i][k'][p']$ is simply an alternative of
$F[i-1][k'-k_i][p'-p_i]$ for every pair $(k_i,p_i)$ such that
$LP_{k_i,p_i}(G_i,L,S_i)$ is feasible, $k_i\le k'$ and $p_i\le p'$.
Thus in polynomial time one can check whether the desired partitions exists, 
and provided that together with a true value we also store the witness
partitions, also find these partitions.
\end{proof}


\section{Rounding}
\label{sec:rounding}

In the previous section we have shown how given a skeleton $S$
one can partition the initial instance into smaller subinstances
with more structural properties.
Our main goal in this section is to show that those structural
properties are in fact sufficient to construct a solution
for each of the subinstances, which is formalized in the following lemma.

\begin{lemma}\label{lem:tra}
Let $I=(G=(\C,\F,E),L,k,p)$ be an instance of \fullsup\  
and let $S \subseteq \F$.
If the following four conditions are satisfied:
\begin{enumerate}[(i)]
  \item $G$ is connected,
  \item\label{it:two} for any $u,u' \in S$, $u \neq u'$ we have $d(u,u') \ge 6$,
  \item\label{it:three} $N^5[S] = \F \cup \C$,
  \item $LP_{k,p}(G,L,S)$ admits a feasible solution,
\end{enumerate}
then one can find a distance-25 solution for $I$ in polynomial time.
\end{lemma}

Before we give a proof of Lemma~\ref{lem:tra},
in Section~\ref{ssec:transfer} we
recall (an adjusted version) of a distance-$r$
transfer, a very useful notion introduced in~\cite{svensson},
together with its main properties.
Next, in Section~\ref{ssec:tree-rounding} we
prove Lemma~\ref{lem:tra}.

\subsection{Distance $r$-transfer}
\label{ssec:transfer}

 \begin{definition}\label{def:tra}
Given a graph $G=(V,E)$ with $W\sub V$, a capacity function $L : W \to \pint$
and $y\in \preal^W$, a vector $y'\in \preal^W$ is a distance-$r$ transfer of
$(G,L,y)$ if
\begin{enumerate}
  \item $\sum_{v\in W} y'_v = \sum_{v\in W} y_v$ and
  \item\label{it:nei} $\sum_{v\in W : d(v,U)\le r} L(v)y'_v \ge \sum_{u\in
  U}L(u)y_u$ for all $U\sub W$.
\end{enumerate}
If $y'$ is a characteristic vector of $F\sub W$, we say that $F$ is an
integral distance-$r$ transfer of $(G,L,y)$.
\end{definition}

Less formally a distance-$r$ transfer is a reassignment,
where the sum of $y$-variables is preserved and 
locally for any set $U \sub W$ the total fractional capacity
in a small neighborhood of $U$ does not decrease.

Like in~\cite{svensson}, an integral distance-$r$
transfer of the fractional solution of the LP already gives a distance-$r+1$
solution (in particular point~\ref{it:nei} of Definition~\ref{def:tra}
ensures that the Hall's condition is satisfied).
The proof must be modified though, so that it encompasses outliers.
\begin{lemma}\label{lem:app}
Let $G=(\C,\F,E)$ be a bipartite graph with a capacity function $L : \F \to
\pint$.
Assume $(x,y)$ is a feasible solution of $LP_{k,p}(G,L,S)$ and $F\sub \F$ is
an integral distance-$r$ transfer of $y$. Then one can find a distance-$r+1$
solution $(C,F,\phi)$ in polynomial time.
\end{lemma}
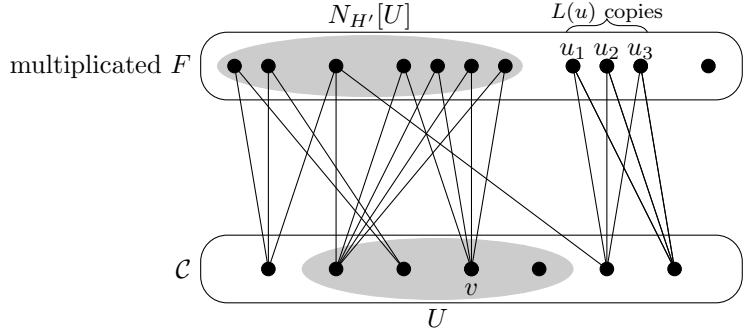
\begin{figure}[ht]
\begin{center}
\begin{tikzpicture}[scale=.9]

\draw[rounded corners=10] (0,3) rectangle (8,4);
\draw[rounded corners=10] (0,0) rectangle (8,1);
\filldraw[black!20] (3.5, 0.5) ellipse (2 and 0.45);
\draw (3.5, 0.05) node[below] {$U$};

\filldraw[black!20] (2.5, 3.5) ellipse (2.25 and 0.45);
\draw (2.5, 4) node[above=-2.5] {$N_{H'}[U]$};

\draw (0,0.5) node[left] {$\mathcal{C}$};
\draw (0, 3.5) node[left] {$\text{multiplicated }F$};

\filldraw (4, .5) circle (0.1) node[below=2] {$v$};

\foreach \x in {1,...,7}{
\filldraw (\x, .5) circle (0.1);
}

\foreach \x in {1,2,4,6,7,8,9,11,12,13,15}{
\filldraw (0.5*\x, 3.5) circle (0.1);
}
.

\draw [decorate,decoration={brace,amplitude=5pt}] (5.4, 4.0) --
node[above]{\footnotesize{$L(u)$ copies}} (6.6, 4.0);

\foreach \x in {1,...,3}{
\filldraw (5+0.5*\x, 3.5) circle (0.1) node[above=-1] {$u_\x$};

}

\foreach \x in {1,...,4}{
\draw (2.5+0.5*\x, 3.5)  -- (4,.5);

\draw (2.5+0.5*\x, 3.5)  -- (2,.5);

}

\draw (2, 3.5)  -- (2,.5);
\draw (1, 3.5)  -- (3,.5);

\draw (1, 3.5)  -- (1,.5);
\draw (.5, 3.5)  -- (1,.5);

\draw (.5, 3.5)  -- (3,.5);
\draw (2, 3.5)  -- (1,.5);

\draw (5.5, 3.5)  -- (6,.5);
\draw (6, 3.5)  -- (6,.5);
\draw (6.5, 3.5)  -- (6,.5);

\draw (2, 3.5)  -- (6,.5);

\draw (6.5, 3.5)  -- (7,.5);
\draw (6, 3.5)  -- (7,.5);
\draw (5.5, 3.5)  -- (7,.5);

\draw (6.5, 3.5)  -- (7,.5);
\draw (6, 3.5)  -- (7,.5);
\draw (5.5, 3.5)  -- (7,.5);

\end{tikzpicture}
\end{center}
\caption{\label{fig:bip}Graph $H'$ obtained from $H$ by removing vertices from
$\F \sm F$ and duplicating each vertex $u\in F$ to its capacity. Shaded ellipses represent
sets used in Hall's theorem.}
\end{figure}

\begin{proof}
Consider a bipartite graph $H=(\C,\F,E_H)$ with $(v,u)\in E_H$ if $d_G(v,u)\le
r+1$. Modify $H$ to obtain $H'$ by removing vertices from $\F \sm F$
and duplicating each vertex $u\in F$ to its capacity, i.e.\ $L(u)$ times, see
also Fig.~\ref{fig:bip}.
Observe that cardinality-$p$ matchings in this graph correspond to
distance-$r+1$ solutions for $G$. 
If any, such a matching can clearly be found in polynomial time.
We shall prove its existence by checking the deficit version of Hall's theorem,
i.e. that for each $U\sub \C$ we have 
$$\sum_{u\in F : d(u,U)\le r+1} L(u) \ge|U|-|\C|+p$$
First, observe that $$\sum_{v\in U, u\in \F} x_{uv} = \sum_{v\in \C, u\in \F}
x_{uv} - \sum_{v\in \C\sm U, u\in F} x_{uv} \overset{(\ref{lp:2}),(\ref{lp:5})}{\ge} p -  \sum_{v\in \C\sm U} 1 = p
- |\C\sm U| = |U|-|\C|+p.$$
Moreover
\begin{multline*}
\sum_{v\in U, u\in \F} x_{uv} = \sum_{v\in U, u\in N_G(U)}
x_{uv}\le \sum_{u\in N_G(U)}\sum_{v\in \C} x_{uv} \overset{(\ref{lp:4})}{\le}
\sum_{u\in N_G(U)} L(u)y_u\\ 
\overset{\textrm{Def.}~\ref{def:tra}\textrm{ point }\ref{it:nei}}{\le}
\sum_{u\in \F : d_G(u, N_G(U))\le r} L(u) = \sum_{u\in \F : d_G(u,U)\le r+1} L(u)\,.
\end{multline*}
Together these equalities conclude the proof.
\end{proof}

We proceed with a pair of simple properties of transfers. 
\begin{fact}\label{fct:comp}
Let $G=(V,E)$ be a graph with $W\sub V$ and a capacity function $L: W \to
\pint$, and let $y,y',y''\in \preal^W$.
Assume $y'$ is a distance-$r$ transfer of $(G,L,y)$
and $y''$ is a distance-$r'$ transfer of $(G,L,y')$. Then $y''$ is a
distance-$r+r'$ transfer of $(G,L,y)$.
\end{fact}
\begin{fact}\label{fct:image}
Let $G=(V,E)$ and $G'=(V',E')$ be graphs with $W\sub V$ and $W\sub V'$ and a
capacity function $L : W\to \pint$. Let $y,y'\in \preal^W$ and let $f : \pint\to
\pint$ be a monotonic function such that $d_G(u,v)\le f(d_{G'}(u,v))$ for any
$u,v\in W$.
Assume $y'$ is a distance-$r$ transfer of $(G',L,y)$. Then $y'$ is a
distance-$f(r)$ transfer of $(G,L,y)$.
\end{fact}

The following is the main technical contribution
of~\cite{svensson}.

\begin{lemma}[\cite{svensson}]\label{lem:tree}
Let $T=(V,E)$ be a tree with a capacity function $L:V\to \pint$ and let $y\in
[0,1]^V$ be a vector such that $y_v=1$ for every non-leaf $v\in V$ and
$\sum_{v\in V} y_v \in \pint$.
Then one can find in polynomial time an integral 
distance-2 transfer of $(T,L,y)$.
\end{lemma}

\subsection{Final rounding}
\label{ssec:tree-rounding}

\begin{lemma}
\label{lem:build-tree}
Let $G=(\C,\F,E)$ be a connected bipartite graph
and let $S \subseteq \F$ such that $d(v,S)\le 5$ for every $v\in \C\cup \F$.
There exists an auxiliary tree $T=(S,E_T)$ 
such that $d(u,u')\le 10$ for any $\{u,u'\}\in E_T$.
Moreover, such a tree can be computed in polynomial time.
\end{lemma}
\begin{proof}
We shall grow a tree adding a leaf in each step.
At the beginning we select any $s\in S$ and initialize with a single-vertex
tree.
Assume we have already grown a tree with vertex-set $S'\sub S$.
Choose a shortest path connecting $S'$ to $S'\sm S$. Such a path exists since
$G$ is connected. If its length is at most 10, we add the endpoint in
$S\sm S'$ to the tree, joining it with the other endpoint.
For a proof by contradiction assume that a shortest path has length greater than
10. Since $G$ is bipartite, its length needs to be even, and thus at least 12.
Choose the midpoint of such a path. Its distance both
to $S'$ and to $S'\sm S$ is at least 6, otherwise the path could be shortened.
This vertex contradicts the assumption that $d(v,S)\le 5$ for every $v\in \C\cup
\F$.
\end{proof}

We are ready to prove Lemma~\ref{lem:tra}.

\begin{proof}[proof of Lemma~\ref{lem:tra}]
Since $G$ is connected and every vertex of $G$ is within distance $5$ from~$S$,
we can use Lemma~\ref{lem:build-tree} to construct a tree $T=(S,E_T)$.
Let us add a duplicate $s'$ of every $s\in S$ to create a bipartite graph
$G'=(\C,\F',E')$, where $\F' = \F \cup S'$ and $S'=\{s' : s\in S\}$.
For each $s\in S$ choose $m_s = \argmax\{ L(u) : u\in  N^2[s]\cap \F\}$ and set  $L(s') = L(m_s)$.
Let us create a tree $T'$ with $V(T')=\F'\sm \{m_s : s\in S\}$.
We build it in two steps, see also Fig.~\ref{fig:tree}:
\begin{enumerate}
  \item create a tree with vertex set $S'$ so that $\{u',v'\}$ is an edge iff
  $\{u,v\}\in E(T)$,
  \item connect each vertex in $\F \sm \{m_s : s\in S\}$ to the closest vertex
  in $S'$.
\end{enumerate}
Observe that endpoints of the edges created in the first step are at most at
distance 10 in $G'$, while endpoints of the edges created in the second step, at
most at distance 4.
Consequently, $d_{G'}(u,v)\le 10 d_{T'}(u,v)$ for any $u,v\in V(T')$.
Moreover, note that all non-leaves of $T'$ belong to $S'$.

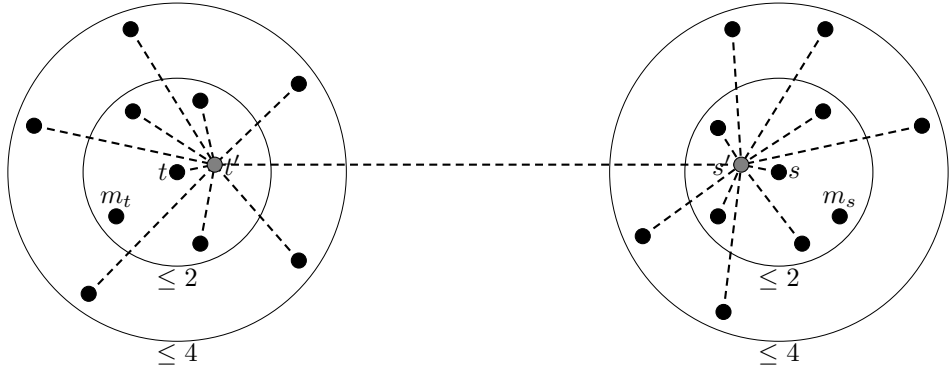
\begin{figure}
\begin{center}
\begin{tikzpicture}[scale=1]
\draw[thick, densely dashed] (3.5, 0.1) -- (-3.5, 0.1);

\begin{scope}[xshift = 4cm, xscale=-1]
\filldraw (0,0) circle (0.1) node[right]{$s$};

\filldraw (6*36:1) circle (0.1) node[above] {$m_s$};

\foreach \x in {1,3.5,7,9}{
\filldraw (36*\x:1) circle (0.1);
\draw[thick, densely dashed] (36*\x:1) -- (.5,.1);
}

\foreach \x in {2,3,4.5,8.1, 9.3}{
\filldraw (36*\x:2) circle (0.1);
\draw[thick, densely dashed] (36*\x:2) -- (.5,.1);
}

\draw[thin] (0,0) circle (1.25);
\draw (0, -1.25) node[below=-2.5] {$\le 2$};
\draw[very thin] (0,0) circle (2.25);
\draw (0, -2.25) node[below=-2.5] {$\le 4$};

\draw[thick, densely dashed] (0, 0) -- (0.5, .1);
\draw[fill=gray] (0.5, 0.1) circle (0.1) node[left] {$s'$};

\end{scope}

\begin{scope}[xshift = -4cm]

\filldraw (0,0) circle (0.1) node[left]{$t$};

\filldraw (6*36:1) circle (0.1) node[above] {$m_t$};

\foreach \x in {2,3.5,8}{
\filldraw (36*\x:1) circle (0.1);
\draw[thick, densely dashed] (36*\x:1) -- (.5,.1);
}

\foreach \x in {1,4.5, 6.5,3,9}{
\filldraw (36*\x:2) circle (0.1);
\draw[thick, densely dashed] (36*\x:2) -- (.5,.1);
}

\draw[thin] (0,0) circle (1.25);
\draw (0, -1.25) node[below=-2.5] {$\le 2$};
\draw[very thin] (0,0) circle (2.25);
\draw (0, -2.25) node[below=-2.5] {$\le 4$};
\draw[thick, densely dashed] (0, 0) -- (0.5, .1);
\draw[fill=gray] (0.5, 0.1) circle (0.1) node[right] {$t'$};
\end{scope}
\end{tikzpicture}
\end{center}
\caption{\label{fig:tree}
A fragment of the tree $T'$ with $s,t\in S$. Nodes of $\F$ are
marked in black, of $S'$ in gray. Edges of $T'$ are represented as dashed lines.
Note that $m_s$ and $m_t$ are not vertices of $T'$.}\
\end{figure}

Let $(x,y)$ be a feasible solution of $LP_{k,p}(G,L,S)$. Note that $y$
can be interpreted as a vector in $\preal^{\F'}$ extending with zeroes at $S'$.
We shall give an integral distance-24 transfer $F$ of $(G',L,y)$.
Despite it being formally a transfer in $G'$, $F$ will be a subset of $\F$, i.e. a
transfer of $(G,L,y)$  as well.

Recall that by (\ref{it:two}), the sets
$N^2[s]$ are pairwise disjoint and in particular $m_s$ are pairwise different. 
This lets us use~\eqref{eq:near} to gather in $s'$ one unit from $N^2[s]$ for
every $s\in S$ so that the whole value in $m_s$
is transferred to $s'$.
Note that $L(s')\ge L(u)$ for each $u\in N^2[s]$, so
this way we obtain a distance-2 transfer $y'$ of $(G',L,y)$. Additionally, we
have made sure that $y'_{m_s}=0$, so $y'$ can be interpreted as a vector in
$\preal^{V(T')}$, and that $y'_{s'}=1$, so $y'$ is 1 for all non-leaves of
$T'$. This lets us use Lemma~\ref{lem:tree} to obtain an integral distance-2 transfer $F'\sub V(T')$ of $(T',L,y')$. According to
Fact~\ref{fct:image} it can be interpreted as a distance-20 transfer of
$(G',L,y')$. Finally we move the value from $s'$ to $m_s$ for each $s\in S$.
Note that these vertices have equal capacities, so this step can be interpreted
as an integral distance-2 transfer.

The final transfer is therefore a composition
of a distance-2 transfer, a distance-20 transfer and a distance-2 transfer.
Thus, by Fact~\ref{fct:comp} it is a distance-24 transfer.\footnote{A simpler
construction gives a distance-$30$ transfer, without introducing additional
vertices $S'$. It is enough first to gather one unit from $N^2[s]$ in $m_s$ and build 
a tree on vertices $m_s$, where adjacent vertices of the tree are at distance
at most $14$ in $G$. By using Lemma~\ref{lem:tree} one obtains a distance-28 transfer,
which together with the initial distance-2 transfer gives an integral distance-30 transfer.}
By Lemma~\ref{lem:app} having an integral distance-24 transfer
is enough to construct a distance-25 solution $\phi$ in polynomial time,
which concludes the proof of Lemma~\ref{lem:tra}.
\end{proof}

\section{Wrap-up}
\label{sec:wrap-up}

With the results of previous section, we are ready to the prove the
main theorem.
\begin{theorem}\label{thm:gen}
The \fullsup\  problem admits a 25-approximation algorithm.
\end{theorem}
\begin{proof}
Section~\ref{sec:thresholding} with Algorithm~\ref{alg:red} provides (a Turing-like)
reduction to graphic instances. 
Algorithm~\ref{alg:pre} of Section~\ref{sec:skeleton} given such an instance outputs several sets. 
Provided that a distance-1 solution exists, one of them is guaranteed to
be a skeleton. Each of these sets is then processed separately. As described
in Section~\ref{sec:clustering}, some redundant vertices are removed and the graph is
partitioned into connected components. Dynamic programming (Lemma~\ref{lem:dp})
is then used to find a compatible partition of $k$ and $p$, 
so that each linear program $LP_{k_i,p_i}(G_i,L,S_i)$ admits a feasible solution.
While this procedure might fail in general, it is guaranteed to succeed for a skeleton, hence at least once if a distance-1 solution exists. 

Note that if such a partition is found, then for each
of the instances $(G_i,L,k_i,p_i)$ together with sets $S_i$,
we can use Lemma~\ref{lem:tra} as all the conditions $(i)-(iv)$ are satisfied.
A sum of solutions for these $\ell$ instances is finally
returned as a distance-25 solution for the original graphic instance.
\end{proof}

\section*{Acknowledgements}

We would like to thank Samir Khuller for suggesting the study of this variant
of the $k$-{\sc center} problem and helpful discussions.

\bibliography{kcenter}

\newpage
\appendix

\section{Soft capacities and uniform capacities}\label{sec:improvements}
\label{app:improvements}
\subsection{Soft capacities}
A variant of \fullsup\ with soft capacities can be easily reduced to the
original problem preserving the quality of solutions.
It suffices to duplicate $|\C|$ times each $v\in \F$.
Opening several facilities in $v$ then corresponds to opening facilities in
several copies of $v$.
\begin{theorem}
The \fullsup\  problem with soft capacities admits a 25-approximation algorithm.
\end{theorem}

\subsection{Uniform capacities}
In the special case of \fullsup\  where the capacities are uniform,
we can obtain a slightly better approximation factor.
Namely, in the proof of Lemma~\ref{lem:tra} 
we can set $m_s=s$ and avoid introducing additional vertices $s'$, using $s$ instead. With this change the
third component of the transfer -- moving the value from $s'$ to $m_s$ -- is not
necessary, thus we get an integral distance-22 transfer. Analogously to
Theorem~\ref{thm:gen}, we then obtain the following result.
\begin{theorem}
The \fullsup\  problem with uniform capacities admits a 23-approximation algorithm.
\end{theorem}

\subsection{Uniform soft capacities}
While we could argue as for general soft capacities that in the case of uniform
soft capacities we have a 23-approximation algorithm, a tailor-made proof gives
much better factor.

It is easy to verify that the ingredients of the proof of
Theorem~\ref{thm:gen} be adapted to soft capacities with two changes:
\begin{itemize}
  \item instead of a set of open facilities, we consider a multiset,
  \item we drop the $y\le \mathbf{1}$ requirement in the LP.
\end{itemize}
Thus, in order to obtain an $r+1$-approximation algorithm it is enough to compute an integral (again, multisets
allowed) distance-$r$ transfer of $y$, where $(x,y)$ is the fractional solution
of the LP for an instance satisfying the conditions of Lemma~\ref{lem:tra}.

Again, we shall start with gathering value from $N^2[s]$ in $s$. This time we
are allowed to gather more than one unit in $s$, so we gather everything from
$N^2[s]$. A vector $y'$ defined this way clearly is a distance-2
transfer of $(G,L,y)$. Moreover, by~\eqref{eq:near} at least one unit is gathered
at each $s\in S$. Like in the proof of Lemma~\ref{lem:tra}, the second component
relies on the structure of $T$.
We connect each $v\in \F\sm S$ to the closest $s\in S$ obtaining a tree $T'$.
This way we have a tree on $\C$ whose non-leaves belong to $S$, and
such that $d_G(u,v)\le 10$ for any $\{u,v\}\in E(T')$.
We shall give an integral distance-1 transfer $y''$ of $(T,L,y')$.
Let us make $T'$ a rooted tree, setting the root at a vertex $r\in S$.
For each $v\in V(T)$ define $Y'_v$ as the sum of $y'_u$ over all descendants $u$ in the
subtree rooted at $v$. 
For each $v\in V(T')$ we transfer $\delta_v := Y'_v - \floor{Y'_v}$ units from $v$ to
its parent $p(v)$.  Note that $Y'_r$ is an integer, since $\sum_{v\in \F}
y'_v = k$, so $\delta_r=0$ and the operation is well defined. 
Observe that for every $v\in V(T')$ it holds that
$$y''_v = y'_v-\delta_v + \sum_{u\;:\:\text{child of }v}\delta_u = \floor{Y'_v}-\sum_{u\;:\; \text{child of }v} \floor{Y'_u}\in \pint.$$
 Also, for any vertex $v$ we have $\delta_v \le y'_v$. That is because for leaves
$\delta_v = Y'_v - \floor{Y'_v} \le Y'_v = y'_v$ and for the remaining vertices $\delta_v = Y'_v -
\floor{Y'_v} \le 1 \le y'_v$, since $v\in S$ so that $y'_v\ge 1$.
Consequently, for any $U\sub \F$, setting $U'=\{u:d_{T'}(u,U)\le 1\}$, we get
$$\sum_{v\in U'}y''_v = \sum_{v\in U'}\left(y'_v-\delta_v + \sum_{u\;:\:\text{child of }v}\delta_u\right)=
\sum_{v\in U'}(y'_v-\delta_v)+\sum_{u:p(u)\in U'}\delta_u \ge\sum_{v\in U}(y'_v-\delta_v)+\sum_{u\in U}\delta_u = \sum_{v\in U}y'_v,$$
since $0\le \delta_v\le y'_v$ for any $v\in \F$. Moreover $L(v)=L$ is a constant, so this inequality 
proves the condition \ref{it:nei}. of Definition~\ref{def:tra}, and thus $y''$ is indeed a distance-1 transfer of $(T,L,y')$.
By Fact~\ref{fct:image} this defines a distance-10 transfer of $(G,L,y')$,
which composed with the previous transfer using Fact~\ref{fct:comp} gives an
integral distance-12 transfer of $(G,L,y)$.
Consequently, repeating the proof of Theorem~\ref{thm:gen} we get the
following result.
\begin{theorem}
The \fullsup\  problem with uniform soft capacities admits a 13-approximation
algorithm.
\end{theorem}

\section{Equivalence of \textsc{Capacitated} $k$-\textsc{supplier with Outliers} and \textsc{Capacitated} $k$-\textsc{center with Outliers}}\label{app:equivalence}

\begin{theorem}
Assume there exists an $r$-approximation algorithm for \textsc{Capacitated} $k$-\textsc{center with Outliers}.
Then there exists an $r$-approximation algorithm for \textsc{Capacitated} $k$-\textsc{supplier with Outliers}.
\end{theorem}
\begin{proof}
Let us consider an instance $I=(\C,\F,d,L, k,p)$ of \textsc{Capacitated} $k$-\textsc{supplier with Outliers}.
Define an instance $I'=(V',d',L', k',p')$ of \textsc{Capacitated} $k$-\textsc{center with Outliers} as
follows: take $V' = (\C\times\{1,\ldots,N\})\cup \F$ where $N=|F|+1$, and
for every $u\in F, v\in C, i\in\{1,\ldots,N\}$ set $d'((v,i),u)=d(v,u)$. Other values of $d'$ 
are taken as the symmetric, transitive closure of those determined explicitly (note that since $d$ was symmetric and satisfied triangle equality,
the closure does not modify any explicitly set value of $d'$).
Also, set $L'(v) = 0$ for $v\in \C\times\{1,\ldots,N\}$, $L'(u)=NL(u)$ for $u\in \F$, $k' = k$, and $p' = pN$.
Clearly $I'$ can be constructed in polynomial time from $I$.
Thus, it suffices to show that a distance-$r$ solution exists in $I$ if
and only if a distance $r$-solution exists in $I'$.

One direction is very simple: assume $\phi:C\to F$ is a distance-$r$ solution in $I$.
Observe that $\phi':(C\times \{1,\ldots,N\})\to F$ defined as $\phi'(v,i)=\phi(v)$
for $v\in C, i\in\{1,\ldots,N\}$ is a distance-$r$ solution in $I'$.

Now, let us prove the other implication. The construction is going
to be similar to the one in the proof of Lemma~\ref{lem:app}.
Assume $\phi':C'\to F$ is a distance-$r$
solution in $I'$. Note that $C'$ may contain vertices from $\F$.
Construct a bipartite graph $H=(\C, \F, E_H)$ with $(v,u)\in E_H$ if $d(v,u)\le r$,
and modify $H$ to obtain $H'$ by removing vertices from $\F\sm F$ and multiplicating each $u\in F$
to its capacity, i.e.\ $L(u)$ times.
Note that $|F|=k'=k$, so  a cardinality-$p$ matching in $H'$ gives a distance-$r$ solution to $I$.
Observe that for any $v\in C$ and $i\in\{1,\ldots,2n\}$, it holds that $d(\phi'(u,i),u)\le r$.
Consequently, for any $U\sub \C$
 we have the following inequality
\begin{multline*}
\sum_{u\in F: d(u,U)\le r} NL(u) \ge \left|\left(U\times\{1,\ldots,N\}\right)\cap C'\right|\ge |C'|-|\F|-\left|(\C\sm U)\times\{1,\ldots,N\}\right|\\
= Np - |\F|+N|U| - N|\C|> N(p+|U|-|\C|-1).
\end{multline*}
Therefore
$$\sum_{u\in F : d(u,U)\le r} L(u) > |U|-|\C|+p-1.$$
Both sides of this inequality are integral, which implies 
$$
\sum_{u\in F : d(u,U)\le r} L(u) \ge|U|-|\C|+p
$$
and, by the deficit version of Hall's theorem, also guarantees the existence of a cardinality $p$-matching in $H'$ and a distance $r$-solution to $I$.
\end{proof}

\section{Connected instance with arbitrarily large integrality gap}\label{app:example}
\begin{fact}
For arbitrarily large $r\in \pint$ there is a graphic instance $I=(G=(\C,\F,E),L,k,p)$ of \fullsup\ 
and a set $S\sub \F$, such that all conditions of Lemma~\ref{lem:tra} except (\ref{it:three}) are satisfied,
but $I$ does not have a distance-$r$ solution.
\end{fact}
\begin{proof}
Assume $r\ge 2$ and fix $N=2r$.
Let $G$ consist of the following components (see also Figure~\ref{fig:ex}): a path of $N+1$ vertices with endpoints $c_1,c_2\in\C$ and inner vertices alternately in $\F$ and $\C$, four vertices $f_{i,j}\in F$ ($i,j\in\{1,2\})$), with $f_{i,j}$ adjacent to $c_i$,
and $12N$ vertices $c_{i,j}\in \C$ ($i\in\{1,2\}, j\in\{1,\ldots,6N\}$), with $c_{i,j}$ adjacent both to $f_{i,1}$ and $f_{i,2}$.
For each $u\in \F$ we set $L(u)=4N$, moreover $k=3$ and $p=12N$. The set $S$ is defined as $\{f_{1,1}, f_{2,1}\}$.

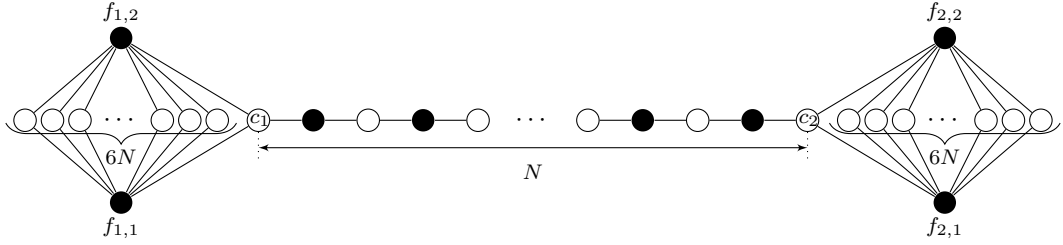
\begin{figure}
\begin{center}

\begin{tikzpicture}[scale=.73,every node/.style={circle, inner sep = 3, }]
\foreach \x in {1,2} {
 \foreach \y in {1,2} {
	\node[fill=black,circle] (f\x\y) at (\x*15,\y*3) {};
}
    \draw (\x*15,3)  node[below=-3] {\footnotesize{$f_{\x,1}$}};
    \draw (\x*15,6)  node[above=-3] {\footnotesize{$f_{\x,2}$}};

}
\foreach \x in {1,2} {
\draw[decorate, decoration={brace, amplitude=8}] (\x*15+2.1,4.45) -- node[below=2]{\footnotesize{$6N$}} ( \x*15-2.1,4.45);
\draw (\x*15, 4.5) node {$\cdots$};
\foreach \y in {0,1,2} {
		\node[draw=black,circle] (c\x\y) at (\x*15-0.75-\y/2,4.5) {};
		\draw (f\x1) -- (c\x\y) -- (f\x2);
}
\foreach \y in {0,1,2} {
		\node[draw=black,circle] (d\x\y) at (\x*15+0.75+\y/2,4.5) {};
		\draw (f\x1) -- (d\x\y) -- (f\x2);
}
}
\foreach \x in {1,2} {
\node[draw=black,circle] (c\x) at (7.5+\x*10, 4.5) {};
\draw (7.5  + \x*10, 4.5) node[left=-9.5] {\footnotesize{$c_{\x}$}};
\draw (f\x1) -- (c\x) -- (f\x2);
}

\foreach \x in {0,2,6,8} {
 \node[fill=black,circle] (a\x) at (18.5+\x, 4.5) {};	
}
\foreach \x in {1,3,5,7} {
 \node[draw=black,circle] (a\x) at (18.5+\x, 4.5) {};	
}
\draw (c1) -- (a0) -- (a1) -- (a2) -- (a3) ;
\draw (c2) -- (a8) -- (a7) -- (a6) -- (a5) ;
\draw (22.5,4.5) node {$\cdots$};

\draw[dotted] (c1) -- (17.5,3.8)  (c2) -- (27.5,3.8);
\draw[arrows={latex'-latex'}] (17.5, 4) -- node[below] {\footnotesize{$N$}} (27.5, 4);

\end{tikzpicture}
\end{center}
\caption{\label{fig:ex}
The graph $G$, vertices in $\C$ are marked as white circles, vertices in $\F$ as black circles.}

\end{figure}

Observe that an instance $I$ constructed this way satisfied conditions of Lemma~\ref{lem:tra} except (\ref{it:three}):
clearly $G$ is connected, $d_G(f_{1,1},f_{2,1})=N+2 \ge 6$. Consider a solution $(x,y)$ of $LP_{k,p}(G,L,S)$ with the following
non-zero coordinates: $y_{f_{i,j}}=\frac{3}{4}$ for $i,j\in\{1,2\}$, $x_{f_{i,j}c_{i,j'}}=\frac{1}{2}$ for $i,j\in\{1,2\}$, $j'\in\{1,\ldots,6N\}$.
It is easy to verify that it is a feasible solution.

It remains to show that $I$ does not have a distance-$r$ solution.
For a proof by contradiction, assume that is does, with $F\sub \F$ being the set of open facilities and $C\sub \C$
being the set of clients served.
Note that each $u\in F$ must serve $4N$ clients, since $p=4Nk$ and $L(u)=4N$ for $u\in \F$.
Let $\F_i = \{u\in \F : d(u,f_{i,1})\le r\}$ and $\C_i = \{v \in \C : d(v, \F_i)\le r\}$ for $i\in\{1,2\}$.
Observe that $\F=\F_1\cup \F_2$ and the sum is disjoint. Consequently $|F\cap \F_i| \ge 2$ and  $|C\cap \C_i|\ge 8N$ for some $i\in\{1,2\}$.
However, $\C_i$ does not contain $c_{3-i,j}$ for any $j\in\{1,\ldots,6N\}$, so $|\C_i|\le |\C|-6N=6.5N+1<8N$, a contradiction.
\end{proof}

\end{document}